\documentclass[11pt]{amsart}

\usepackage{amssymb}
\usepackage{epsfig}
\usepackage{comment}
\usepackage{amsmath}
\usepackage{subcaption}
\usepackage[section]{placeins}
\usepackage{mathrsfs}
\usepackage{graphicx}
\usepackage[notrig]{physics}
\usepackage{units}
\usepackage{algorithm,algorithmic}

\newtheorem{thm}{Theorem}[section]
\newtheorem{prop}[thm]{Proposition}

\theoremstyle{definition}

\theoremstyle{remark}

    \newcommand{\mbb}[1]{\mathbb{#1}}

\begin{document}

\title[Data-Driven viscoelasticity]{Model-Free Data-Driven viscoelasticity in the frequency domain}

\author
{
    Hossein Salahshoor${}^1$, Michael~Ortiz${}^{1,2}$
}

\address
{
    ${}^1$Division of Engineering and Applied Science,
    California Institute of Technology, 1200 E.~California Blvd., Pasadena, CA 91125.
    ${}^2$Hausdorff Center for Mathematics, Universit\"at Bonn, Endenicher Allee 60, 53115 Bonn, Germany \\ \newline\indent
}

\email{ortiz@caltech.edu}
%\urladdr{http://www.ortiz.caltech.edu/$\sim$ortiz/home.shtml}

\begin{abstract}
We develop a Data-Driven framework for the simulation of wave propagation in viscoelastic solids directly from dynamic testing material data, including data from Dynamic Mechanical Analysis (DMA), nano-indentation, Dynamic Shear Testing (DST) and Magnetic Resonance Elastography (MRE), without the need for regression or material modeling. The problem is formulated in the frequency domain and the method of solution seeks to minimize a distance between physically admissible histories of stress and strain, in the sense of compatibility and equilibrium, and the material data. We metrize the space of histories by means of the flat-norm of their Fourier transform, which allows consideration of infinite wave trains such as harmonic functions. Another significant advantage of the flat norm is that it allows the response of the system at one frequency to be inferred from data at nearby frequencies. We demonstrate and verify the approach by means of two test cases, a polymeric truss structure characterized by DMA data and a 3D soft gel sample characterized by MRE data. The examples demonstrate the ease of implementation of the Data-Driven scheme within conventional commercial codes and its robust convergence properties, both with respect to the solver and the data.
\end{abstract}

\maketitle

%\tableofcontents

\section{Introduction}

A number of techniques are presently available for characterizing viscoelastic materials in the frequency domain, including Dynamic Mechanical Analysis (DMA) \cite{PAN21}, nanoindentation \cite{herbert2008nanoindentation, herbert2009measuring}, dynamic shear testing (DST) \cite{bayly2008magnetic, arbogast1998material} and magnetic resonance elastography (MRE) \cite{MLRGME95}, among others. The systematic application of these techniques results in large material data sets that challenge analytical representation and computational utilization. Two main strategies arise in response to this wealth of material data: i) Modeling the data using conventional rheological materials or, more recently, neural-network representations and machine learning regression \cite{Xu:2021}, then using the calibrated material models in the solution of boundary-value problems; ii) Using the data directly in the solution of suitably-formulated boundary-value problems in the spirit of Data-Driven (DD) mechanics \cite{KO16}, thereby bypassing the material-modeling step altogether. The main advantage of the DD paradigm is that is makes direct use of {\sl the data, all the data, and nothing but the data.} In particular, no loss of information is incurred with respect to the material data set, no modeling biases or errors are introduced and no {\sl ad hoc} choices need to be made regarding the functional form of the material model. In addition, if material behavior is properly sampled with increasing fidelity, the DD solution may be expected to converge to the solution of the underlying, and unknown, material response \cite{Conti:2018}.

The application of the DD paradigm to history-dependent materials has been investigated in \cite{Eggers:2019}. The history dependence of the material response requires characterizing entire state histories $(\epsilon(x,t), \sigma(x,t))$ of stress and strain at material points $x$ in the solid. The collection of all such state histories defines a material-history data set $\mathcal{D}$ that characterizes the material response of the solid. In addition, the state histories are subject to the constraints of compatibility and equilibrium, both in the interior and at the boundary. The collection of all histories satisfying those physical constraints defines a constraint set $\mathcal{E}$ of admissible state histories. Evidently, the actual histories undergone by the solid are those that lie at the intersection of $\mathcal{D}$ and $\mathcal{E}$, i.~e., histories that are physically admissible, in the sense of compatibility and equilibrium, and are contained in the material-history data set characterizing material behavior.

In this work, we consider the case, often encountered in practice, in which the full material behavior of a viscoelastic material is not known, but only partially characterized by empirical material-history data sets $\mathcal{D}_h$. For instance, such data sets could be the result of DMA \cite{PAN21}, in which case the sampled histories $(\epsilon(x,t), \sigma(x,t))$ are harmonic in time, or could be the result of creep and relaxation tests, or computed from micromechanics \cite{Heyden:2016, Manav:2021}, or determined otherwise. In general, the intersection of the empirical material-history data set $\mathcal{D}_h$ and the constraint set $\mathcal{E}$ is likely to be empty, i.~e., no material history recorded in $\mathcal{D}_h$ is likely to be compatible and in equilibrium with the applied loading history. In order to circumvent this difficulty, we relax the notion of solution and seek to determine admissible state histories in $E$ that are closest, in the sense of a suitable metric, to the empirical material-history data sets $\mathcal{D}_h$. In this manner, we expect to identify approximating state histories $(\epsilon_h(x,t),\sigma_h(x,t))$ that converge to the exact state history $(\epsilon(x,t),\sigma(x,t))$ as the sequence $\mathcal{D}_h$ of sets converges to $\mathcal{D}$ in some appropriate manner.

Here, we focus on questions of practical implementation of DD linear viscoelasticity and demonstrate convergence with respect to material data by means of selected numerical tests. A first question of crucial importance concerns the choice of metric used to select approximate state histories. We specifically consider empirical material-history data sets $\mathcal{D}_h$ generated by frequency-domain characterization techniques such as DMA and MRE. In particular, the set $\mathcal{D}_h$ contains uniform harmonic state histories $(\epsilon(t),\sigma(t))$ defined for all times. Evidently, integral norms fail to be defined for such histories and cannot be taken as a basis for DD analysis. Instead, we note that the Fourier transforms $(\hat{\epsilon}(\omega), \hat{\sigma}(\omega))$ of the sampled state histories are Dirac-deltas concentrated at the applied frequency $\omega$, with complex amplitudes related by a complex modulus $\mathbb{E}$. A natural empirical representation of such behavior is by means of material data sets $\mathcal{D}_h$ consisting of pairs $(\omega_i,\mathbb{E}_i)$ of applied frequencies $\omega_i$ and corresponding complex moduli $\mathbb{E}_i$. In addition, a canonical metric for computing the distance between two such data points is supplied by the {\sl flat norm} \cite{MO04, CV08}, which measures the discrepancy in both amplitude and frequency in the Fourier representation. A crucial property of the flat norm is that it allows the response of the material at one frequency to be compared to data at nearby frequencies. This property is particularly useful when dealing with sparse material data sets $\mathcal{D}_h$ containing data at a finite, possibly small, number of applied frequencies.

The article is organized as follows. Following an introductory section aimed at defining the problem and establishing notation, Section~\ref{06F2wm}, an implementation of the frequency-domain DD linear-viscoelastic formalism is presented in Section~\ref{eM62eR}. In Section~\ref{p1PDI0}, we illustrate the DD approach by means of two examples of application: i) A truss structure made of a polymer characterized by DMA data, representative of polymer-based additive manufacturing technology \cite{Arefin:2021}; and ii) a three-dimensional viscoelastic gel characterized by MRE data, commonly used as phantom in the field of ultrasonic biomechanics \cite{Rabut:2020}. In both cases, we demonstrate the convergence of the DD approach with respect to the data. We conclude with a summary and outlook in Section \ref{Br1zak}.

\section{Dynamic viscoelasticity}\label{06F2wm}

For definiteness, we restrict attention to finite-dimensional discrete models. The models under consideration comprise $m$ material points, or structural members in the setting of structural mechanics, whose state at time $t$ is characterized by strains $\epsilon(t) \equiv \{\epsilon_e(t) \in \mathbb{R}^{d},\ e=1,\dots,m\}$ and stresses $\sigma(t) \equiv \{\sigma_e(t) \in \mathbb{R}^{d},\ e=1,\dots,m\}$, where $d$ is the stress and strain dimension.

\subsection{Governing equations}

The state of the system is subject to compatibility and dynamic equilibrium constraints of the form
\begin{subequations}\label{9qHWzU}
\begin{align}
    &
    \epsilon_e(t) = B_e u(t) + g_e(t), \quad e = 1,\dots m ,
    \\ &
    M \ddot{u}(t) + \sum_{e=1}^m w_e B_e^T \sigma_e(t) = f(t) ,
\end{align}
\end{subequations}
where $u(t) \in \mathbb{R}^{n}$ is the array of unconstrained nodal displacements at time $t$, with $n$ the spatial dimension, $f(t) \in \mathbb{R}^n$ is a nodal force array at time $t$, $g_e(t) \in \mathbb{R}^{d}$ is an initial strain at time $t$ resulting from prescribed displacement boundary conditions, $w_e$ are positive weights, $B_e \in L(\mathbb{R}^n, \mathbb{R}^{d})$ is the strain-displacement matrix for material point $e$ and $M \in \mathbb{R}^{n \times n}$ is a mass matrix. The constraints (\ref{9qHWzU}) can also be expressed in compact form as
\begin{subequations}\label{vZftZT}
\begin{align}
    &
    \epsilon(t) = B u(t) + g(t) ,
    \\ &
    M \ddot{u}(t) + B^T W \sigma(t) = f(t)
\end{align}
\end{subequations}
with $B = (B_1, \dots, B_m) \in \mathbb{R}^{N \times n}$, $N = m d$, $\epsilon(t) = (\epsilon_1(t), \dots, \epsilon_m(t))  \in \mathbb{R}^N$, $g(t) = (g_1(t), \dots, g_m(t))  \in \mathbb{R}^N$, $\sigma(t) = (\sigma_1(t),\dots,$ $\sigma_m(t))$  $\in$ $\mathbb{R}^N$ and $W \in \mathbb{R}^{N\times N}$ is a diagonal matrix of weights such that $W \sigma = \{w_1 \sigma_1,\dots, w_m \sigma_m\}$.

Classically, problem (\ref{9qHWzU}) is closed by assuming a hereditary dependence
%\begin{equation}\label{5wb8D7}
%    \sigma_e = F_e(\epsilon_e) ,
%\end{equation}
of the local stress history $\sigma_e(t)$ on the local strain history $\epsilon_e(t)$ at every material point $e$. The mapping $F_e$ maps local histories of strain to local histories of stress and must be {\sl causal}, i.~e., $\sigma_e'(t) = \sigma_e''(t)$ if $\sigma_e'=F_e(\epsilon_e')$, $\sigma_e''=F_e(\epsilon_e'')$, and $\epsilon_e'(s) = \epsilon_e''(s)$ for all $s \leq t$.

\subsection{Fourier representation}

Dynamic viscoelasticity techniques, such as Dynamic Mechanical Analysis (DMA) \cite{PAN21}, are common means of characterizing viscoelastic behavior as a function of frequency and temperature. DMA consists of applying a harmonic force to a test specimen and measuring the resulting strain response. Evidently, DMA contemplates histories of strain and stress that extend over infinite time. In order to account for such histories, we adopt a Fourier representation. Formally, an application of the Fourier transform to (\ref{vZftZT}) gives
\begin{subequations}\label{2viLPa}
\begin{align}
    & \label{tQ8qYF}
    \hat{\epsilon}(\omega) = B \hat{u}(\omega) + \hat{g}(\omega) ,
    \\ &
    B^T W \hat{\sigma}(\omega) - M \omega^2 \hat{u}(\omega) = \hat{f}(\omega) ,
\end{align}
\end{subequations}
where $\hat{}$ denotes the Fourier transform of a function and $\omega$ is the frequency. We note that the Fourier transforms of infinite wave trains can be general measures. For instance, the Fourier transform of a harmonic time history is a Dirac measure centered at the harmonic frequency. Therefore, in general the relations (\ref{2viLPa}) must be interpreted in the sense of measures.

\subsection{Linear viscoelasticity}

At every material point $e$, linear viscoelasticity is characterized as a linear {\sl hereditary law} of the form
\begin{equation}\label{j1eDro}
    \sigma_e(t)
    =
    \int_{-\infty}^t
        \mathbb{J}_e(t-s) \, \dot{\epsilon}_e(s)
    \, ds
    =
    (\mathbb{J}_e * \dot{\epsilon}_e)(t)    ,
\end{equation}
where $\mathbb{J}_e : [0,+\infty) \to L(\mathbb{R}^{{d}\times{d}}_{\rm sym}, \mathbb{R}^{{d}\times{d}}_{\rm sym})$, $\mathbb{J}_e^T(t) = \mathbb{J}_e(t)$, $\mathbb{J}_e(t)$ to $0$ for $t < 0$, is the relaxation function and $*$ is the convolution operator. We may regard (\ref{j1eDro}) as a mapping that gives the stress $\sigma(t)$ and time $t$ as a function of the past history of strain. Evidently, (\ref{j1eDro}) satisfies causality by construction.

Taking the Fourier transform of (\ref{j1eDro}) and applying the convolution theorem, we obtain
\begin{equation}\label{22PiAI}
    \hat{\sigma}_e(\omega)
    =
    i \omega
    \hat{\mathbb{J}}_e(\omega) \hat{\epsilon}_e(\omega)
    \equiv
    \mathbb{E}(\omega) \hat{\epsilon}_e(\omega) ,
\end{equation}
where $\mathbb{E}(\omega) = \mathbb{E}'(\omega) + i \mathbb{E}''(\omega)$ is the complex modulus and
\begin{subequations} \label{CM8rWj}
\begin{align}
    &
    \mathbb{E}'(\omega)
    =
    \omega
    \int_0^\infty \mathbb{J}(s) \sin(\omega s) \, ds ,
    \\ &
    \mathbb{E}''(\omega)
    =
    \omega
    \int_0^\infty \mathbb{J}(s) \cos(\omega s) \, ds ,
\end{align}
\end{subequations}
are the storage and loss moduli, respectively. We note that the relation (\ref{CM8rWj}) separates by frequency in the linear case. These moduli satisfy the
{\sl Kramers-Kronig relations}
\begin{equation}
    \mathbb{E}'(\omega)
    =
    \frac{2}{\pi}
    \mathcal{P} \int_0^\infty
    \frac{\bar{\omega} \mathbb{E}''(\bar{\omega})}{\bar{\omega}^2 - \omega^2}
    \, d\bar{\omega},
    \qquad
    \mathbb{E}''(\omega)
    =
    -
    \frac{2\omega}{\pi}
    \mathcal{P} \int_0^\infty
    \frac{\mathbb{E}'(\bar{\omega})}{\bar{\omega}^2 - \omega^2}
    \, d\bar{\omega},
\end{equation}
which are necessary and sufficient for causality.

Inserting the viscoelastic relation (\ref{22PiAI}) into (\ref{2viLPa}) and eliminating $\hat{\epsilon}(\omega)$ and $\hat{\sigma}(\omega)$ gives the displacement problem
\begin{equation}\label{B9Jhtb}
     B^T W \mathbb{E}(\omega) (B \hat{u}(\omega) + \hat{g}(\omega))
     -
     M \omega^2 \hat{u}(\omega)
     =
     \hat{f}(\omega) ,
\end{equation}
to be satisfied for every frequency $\omega$. Once $\hat{u}(\omega)$ is known, $\hat{\epsilon}(\omega)$ and $\hat{\sigma}(\omega)$ follow from (\ref{tQ8qYF}) and (\ref{22PiAI}).

\subsection{Steady state under monochromatic transduction}

As a particular example of application, suppose that the forcing is harmonic,
\begin{equation}
    \hat{f} = 2 \pi F \delta_{\Omega},
    \quad
    \hat{g} = 2 \pi G \delta_{\Omega},
\end{equation}
where $F\in\mathbb{C}^n$ and $G\in\mathbb{C}^N$ are constant complex amplitudes and $\Omega$ is the transduction frequency. Harmonic forcing of this type is induced, e.~g., as a result of monochromatic transduction. Then, we may solve (\ref{2viLPa}) by setting
\begin{equation}
    \hat{u} = 2 \pi U \delta_{\Omega},
    \quad
    \hat{\epsilon} = 2 \pi E \delta_{\Omega},
    \quad
    \hat{\sigma} = 2 \pi S \delta_{\Omega},
\end{equation}
with $U\in\mathbb{C}^n$, $E\in\mathbb{C}^N$ and $S\in\mathbb{C}^N$ also constant complex amplitudes such that
\begin{subequations}\label{Js6jNy}
\begin{align}
    &\label{xDK1CK}
    E = B U + G ,
    \\ &
    B^T W S - M \Omega^2 U = F .
\end{align}
\end{subequations}
Suppose that the complex modulus $\mathbb{E}(\Omega)$ of the material is known exactly. Then,
\begin{equation}\label{z2VFiv}
    S = \mathbb{E}(\Omega) \, E ,
\end{equation}
and (\ref{Js6jNy}) reduces to
\begin{equation}\label{5Vx1Hp}
    B^T W \mathbb{E}(\Omega) (B U + G)- M \Omega^2 U = F ,
\end{equation}
which is a monochromatic form of (\ref{B9Jhtb}) and can be solved directly for the steady-state displacement amplitude.

\section{Data-Driven reformulation}\label{eM62eR}

An essential difficulty with the classical approach just described is that, whereas the field equations (\ref{2viLPa}) are known exactly and are uncertainty-free, the material law (\ref{22PiAI}) is known only empirically through experimental data. We wish to reformulate the problem in such a way that predictions can be made directly from the field equations and the available empirical data. We additionally wish such Data-Driven (DD) predictions to be approximating in the sense of converging to the classical solution when the empirical data converges, in some appropriate sense to be made precise, to an underlying---and unknown---material law of the form (\ref{22PiAI}).

\subsection{The data-driven problem with fully-known material behavior}

Let $\mathcal{Z}$ be the phase space of all possible histories $(\epsilon(t), \sigma(t))$. We may then regard the local constitutive relation (\ref{j1eDro}), or its Fourier representation (\ref{22PiAI}), as restricting the possible histories to a material set
\begin{equation}\label{kST5C3}
    \mathcal{D} = \{(\epsilon(t), \sigma(t)) \in \mathcal{Z} \, : \,
    \text{eq.}~(\ref{j1eDro}) \text{ or } \text{eq.}~(\ref{22PiAI}) \} .
\end{equation}
Likewise, we may regard (\ref{vZftZT}), or its Fourier representation (\ref{2viLPa}) as defining a constraint set of admissible histories
\begin{equation}
    \mathcal{E} = \{(\epsilon(t), \sigma(t)) \in \mathcal{Z} \, : \,
    \text{eq.}~(\ref{vZftZT}) \text{ or } \text{eq.}~(\ref{2viLPa}) \} .
\end{equation}
The aim then is to find all histories $(\epsilon(t), \sigma(t))$ in the intersection $\mathcal{D} \cap \mathcal{E}$, i.~e., all histories that are admissible and satisfy the constitutive relations of the material.

\subsection{Empirical data}

Suppose that the complex modulus $\mathbb{E}(\omega)$ of a linear viscoelastic material, and therefore its data set $\mathcal{D}$, eq.~(\ref{kST5C3}), is not known exactly but only approximately through empirical data. In the present work, we specifically consider local material data of the Dynamical Material Analysis (DMA) type. Thus, the behavior of every material point $e$ is characterized by a set $\mathcal{P} = \{(\omega_i,\mathbb{E}_i),\ i=1,\dots,N\}$, consisting of measured pairs of frequencies and corresponding complex moduli. Alternatively, we may regard the data as a collection of experimentally observed harmonic histories of stress and strain
\begin{equation}
    \epsilon_e(t) = E {\rm e}^{i\omega_i t},
    \quad
    \sigma_e(t) = S {\rm e}^{i\omega_i t},
    \quad
    S = \mathbb{E}_i \, E ,
    \quad
    i = 1,\dots, N .
\end{equation}
Collecting these histories, we obtain a local material data set of the form
\begin{equation}
\begin{split}
    &
    \mathcal{D}_{\rm loc}
    =
    \{
        (\epsilon_e(t),\sigma_e(t)) \in \mathcal{Z}_{\rm loc} \, : \,
        \\ & \qquad\qquad
        \hat{\epsilon}_e = \sum_{i=1}^{N} 2 \pi E \delta_{\omega_i},
        \ \hat{\sigma}_e = \sum_{i=1}^{N} 2 \pi S \delta_{\omega_i},
        \ S = \mathbb{E}_i \, E
    \} ,
\end{split}
\end{equation}
where $\mathcal{Z}_{\rm loc}$ denotes the phase space of local histories $(\epsilon_e(t),\sigma_e(t))$ of material point $e$. The corresponding global material data set is, then,
\begin{equation}
    \mathcal{D}
    =
    \mathcal{D}_{\rm loc}^m
    =
    \{
        (\epsilon(t),\sigma(t)) \in \mathcal{Z} \, : \,
        (\epsilon_e(t),\sigma_e(t)) \in \mathcal{D}_{\rm loc} ,
        \ e = 1,\dots,m
    \} .
\end{equation}

\subsection{The data-driven problem with empirical material data}

When the material behavior is characterized by means of an empirical data set $\mathcal{D}$ such as just described, the intersection $\mathcal{D} \cap \mathcal{E}$ is likely to be empty, i.~e., there is no history $(\epsilon(t),\sigma(t))$ that is both admissible and in the material data set, and the notion of solution needs to be extended. Following \cite{Eggers:2019}, we seek instead a pair of histories, an admissible history $(\epsilon'(t),\sigma'(t)) \in \mathcal{E}$ and a material history $(\epsilon''(t),\sigma''(t)) \in \mathcal{D}$, that are as close as possible to each other.

A suitable notion of distance between dynamical histories, including infinite wave trains whose Fourier transforms are general measures, is supplied by the flat norm \cite{MO04}. For present purposes, it suffices to characterize the flat-norm distance between harmonic histories. To this end, we begin by introducing a norm in the complex stress-strain space $\mathbb{C}^d\times\mathbb{C}^d$
\begin{equation}
    \| Z_e \|^2
    =
    C \| E_e \|^2 + C^{-1} \| S_e \|^2 ,
\end{equation}
where we write $Z_e = (E_e, S_e)$, with $E_e \in \mathbb{C}^d$ and $S_e \in \mathbb{C}^d$ for short and $C>0$ is a parameter introduced to even out dimensions.

Consider now two local harmonic histories $z_e'(t)$ and $z_e''(t)$ with Fourier representation
\begin{equation}
    \hat{z}'_e = 2 \pi Z'_e \delta_{\omega'} ,
    \qquad
    \hat{z}''_e = 2 \pi Z''_e \delta_{\omega''} ,
\end{equation}
where $Z'_e$, $Z''_e \in \mathbb{C}^d\times\mathbb{C}^d$ are constant complex amplitudes and $\delta_{\omega'}$ and $\delta_{\omega''}$ are Dirac measures centered at the harmonic frequencies $\omega'$ and $\omega''$, respectively. Then, we identity the distance between $z'(t)$ and $z''(t)$ with the flat-norm distance $\| \hat{z}'_e - \hat{z}''_e \|$ between $\hat{z}'_e$ and $\hat{z}''_e$, which is given explicitly by Prop.~\ref{iWC5ll} of Appendix~\ref{af7qTx} with $\lambda = 1/\omega_{\rm cutoff}$ and $\omega_{\rm cutoff}$ a cutoff frequency. As may be seen from Prop.~\ref{iWC5ll}, the flat-norm distance between two Dirac measures quantifies the difference between their amplitudes as well as the distance between their points of application, which indeed defines a natural distance between Dirac measures. In particular the convergence $\hat{z}'_e \to \hat{z}''_e$ requires their amplitudes and their points of application to converge, i.~e., $Z'_e \to Z''_e$ in $\mathbb{C}^d\times\mathbb{C}^d$ and $\omega' \to \omega''$ in $\mathbb{R}$, as expected.

We can further extend the flat norm distance to two global harmonic histories,
\begin{equation}
    \hat{z}' = 2 \pi Z' \delta_{\omega'} ,
    \qquad
    \hat{z}'' = 2 \pi Z'' \delta_{\omega''} ,
\end{equation}
with $Z'$, $Z'' \in \mathbb{C}^N\times\mathbb{C}^N$, in a natural way as
\begin{equation}
    \| z'(t) - z''(t) \|
    =
    \sum_{e=1}^m
    \| z'_e(t) - z''_e(t) \| .
\end{equation}
The corresponding DD problem is, then,
\begin{equation}\label{eE5sIt}
    \inf_{y(t)\in\mathcal{D}}
    \inf_{z(t)\in\mathcal{E}}
    \| y(t) - z(t) \| ,
\end{equation}
i.~e., we seek a material history $y(t)$ and an admissible history $z(t)$ that are closest to each other in the sense of the flat norm.

We may regard (\ref{eE5sIt}) as an adversarial game with players $y(t)$ and $z(t)$ with goals
\begin{subequations}\label{8x8Jfh}
\begin{align}
    &
    F(y(t),z(t)) = \| y(t) - z(t) \| + I_{\mathcal{D}}(y(t)) ,
    \\ &
    G(y(t),z(t)) = \| y(t) - z(t) \| + I_{\mathcal{E}}(z(t)) ,
\end{align}
\end{subequations}
where $I_{\mathcal{D}}$ and $I_{\mathcal{E}}$ are the indicator functions of $\mathcal{D}$ and $\mathcal{E}$, respectively, i.~e., they are zero in $\mathcal{D}$ and $\mathcal{E}$, respectively, and $+\infty$ elsewhere. In the game, $y(t)$ plays to minimize $F(\cdot,z(t))$, for given $z(t)$, whereas $y(t)$ plays to minimize $G(y(t),\cdot)$, for given $y(t)$. Thus, the goal of $y(t)$ and $z(t)$ is to remain in $\mathcal{D}$ and $\mathcal{E}$, respectively, while simultaneously minimizing their distance, as required.

Apart from being natural for infinite wave trains such as harmonic histories, the flat norm distance affords the great practical advantage of enabling the use of data at frequencies other that the applied frequencies of excitation, with an appropriate penalty accorded to the discrepancy between the two. This property is specially advantageous when the data is sampled at discrete frequencies, which are likely to differ from the applied frequency.

\begin{algorithm}[ht]
\caption{Data-driven solver, harmonic loading.} \label{9WfJmZ}
\begin{algorithmic}

\REQUIRE Dynamic viscoelasticity data set $\mathcal{P}$, $B$-matrix, member weights $W$. Applied amplitudes $F$ and $G$, applied frequency $\Omega$, cutoff frequency $\omega_{\rm cutoff}$.
\STATE i) {\bf Initialization}. Set $k=0$. Initial local data assignment:
\FORALL {$e=1,\dots,m$}
\STATE Choose $(\omega^{(0)}_e,\mathbb{E}^{(0)}_e)$ randomly from $\mathcal{P}$.
\ENDFOR
\STATE ii) {\bf Displacement problem}. Solve for $U^{(k)}$ such that
\begin{equation}
    B^T W \mathbb{E}^{(k)} (B U^{(k)} + G)- M \Omega^2 U^{(k)} = F ,
\end{equation}
\STATE iii) {\bf Local states}. Compute
\begin{equation}
    E^{(k)} = B U^{(k)} + G , \quad S^{(k)} = \mathbb{E}^{(k)} E^{(k)} .
\end{equation}
\STATE iv) {\bf Data assignment}.
\FORALL {$e=1,\dots,m$}
\STATE iv.a) Find $(\omega_i,\mathbb{E}_i)$ in $\mathcal{P}$ s.~t.~$d(Z^{(k)}_e \delta_\Omega,\,\{S=\mathbb{E}_i E\}\delta_{\omega_i})$ is minimal.
\STATE iv.b) Set $(\omega^{(k+1)}_e,\mathbb{E}^{(k+1)}_e)= (\omega_i,\mathbb{E}_i)$
\ENDFOR
\STATE v) {\bf Convergence test}.
\IF{$\mathbb{E}^{(k+1)}=\mathbb{E}^{(k)}$}
\STATE $U=U^{(k)}$, $E=E^{(k)}$, $S=S^{(k)}$, {\bf exit}.
\ELSE
\STATE $k \leftarrow k+1$, {\bf goto} (ii).
\ENDIF
\end{algorithmic}
\end{algorithm}

\section{Numerical examples} \label{p1PDI0}

In this section, we present two examples of application that illustrate the range and scope of the Data-Driven (DD) framework developed in the foregoing. We place particular focus on the convergence of the scheme with respect to data. Specifically, we consider a sequence $\mathcal{P}_h$ of data sets providing an increasingly better representation of the exact complex modulus $\mathbb{E}(\omega)$ and verify that, indeed, the DD solutions $(y_h(t)), z_h(t))$ converge to the exact solution $y(t)=z(t)$.

\subsection{Fixed-point iteration}

We solve the DD problem (\ref{eE5sIt}), or the equivalent game (\ref{8x8Jfh}), by means of a fixed-point iteration adapted from \cite{KO16}. We assume that the material is characterized by means of discrete data of the form $\mathcal{P} = \{(\omega_i,\mathbb{E}_i),\ i=1,\dots,N\}$. The solver is summarised in Algorithm~\ref{9WfJmZ}.

\subsection{3D printed polymeric structures}

\begin{figure}[ht]
\begin{center}
	\begin{subfigure}{0.445\textwidth}\caption{} \includegraphics[width=0.99\linewidth]{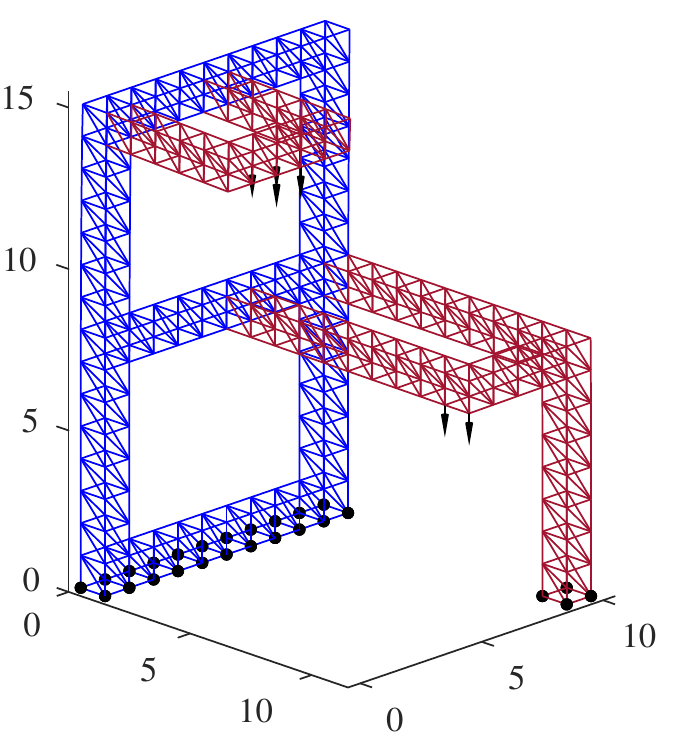}
	\end{subfigure}
	\begin{subfigure}{0.545\textwidth}\caption{} \includegraphics[width=0.99\linewidth]{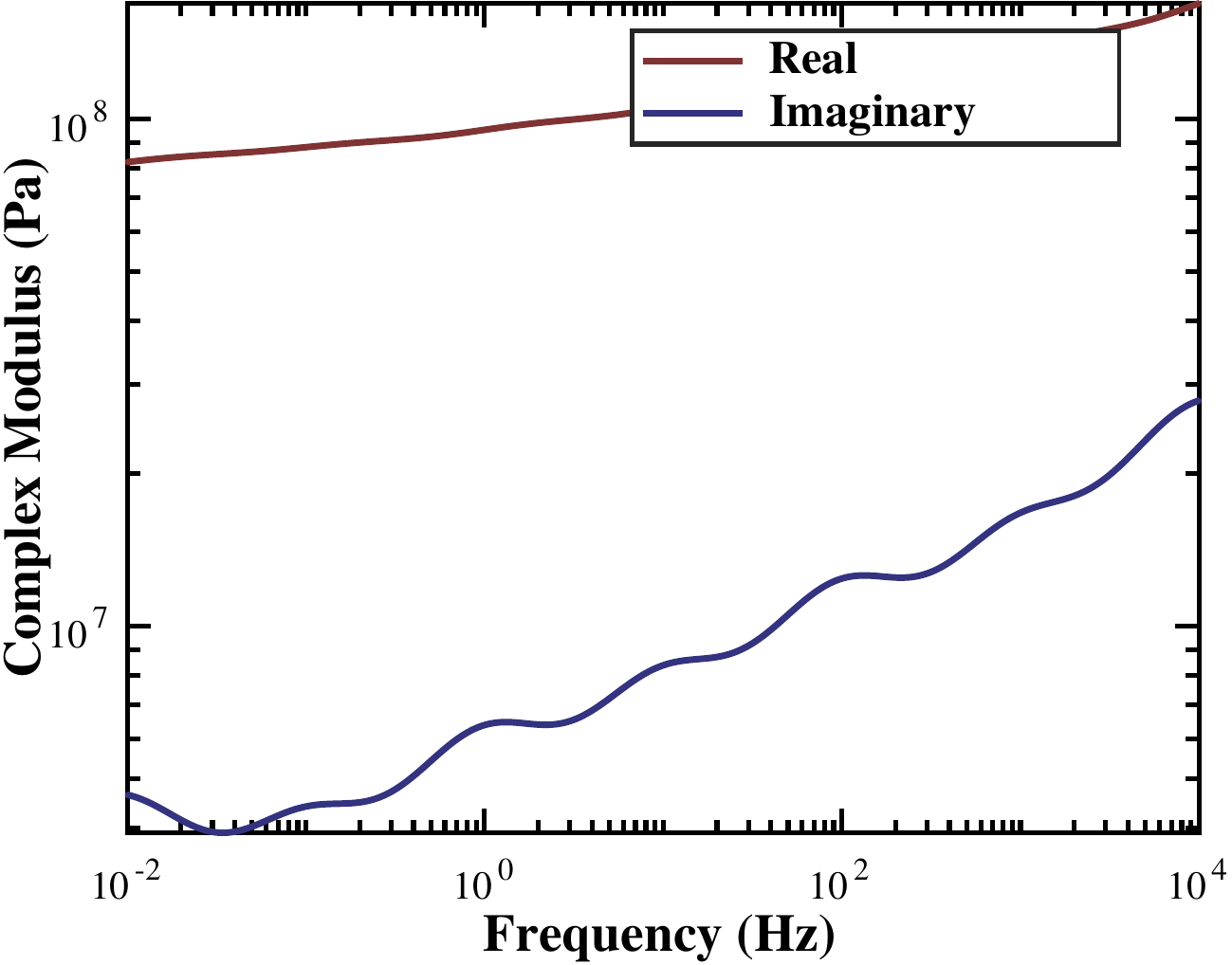}
	\end{subfigure}
    \caption{3D-printed polymeric truss. a) Truss geometry, supports (black dots) and prescribed harmonic displacements (arrows). b) Reference material model for polyurea ( adapted from \cite{knauss2007improved}).} \label{T05Xnh}
\end{center}
\end{figure}

\begin{figure}[ht]
\begin{center}
	\begin{subfigure}{0.45\textwidth}\caption{} \includegraphics[width=0.99\linewidth]{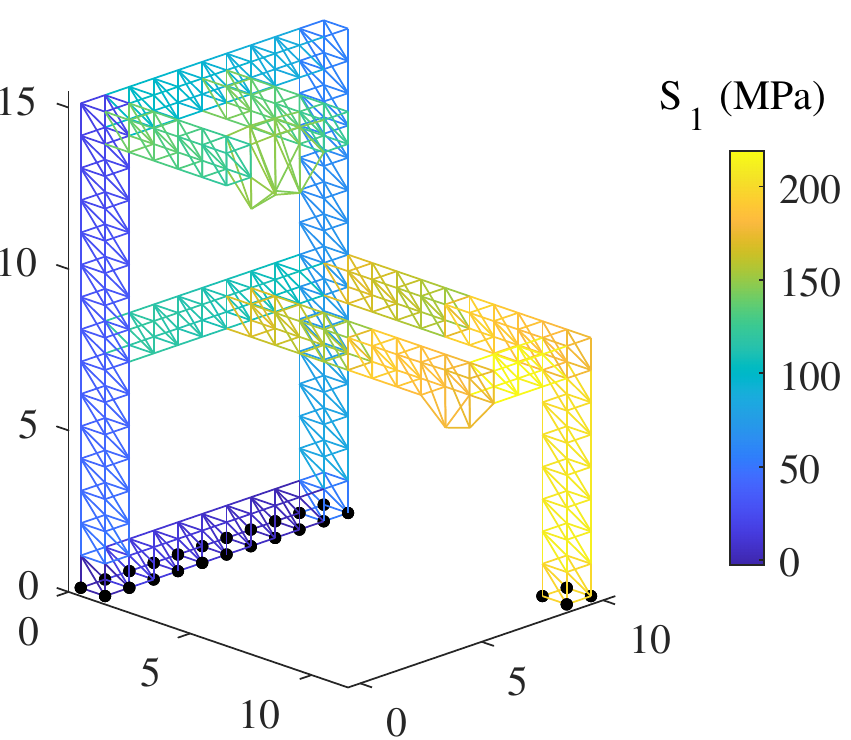}
	\end{subfigure}
$\quad$
	\begin{subfigure}{0.45\textwidth}\caption{} \includegraphics[width=0.99\linewidth]{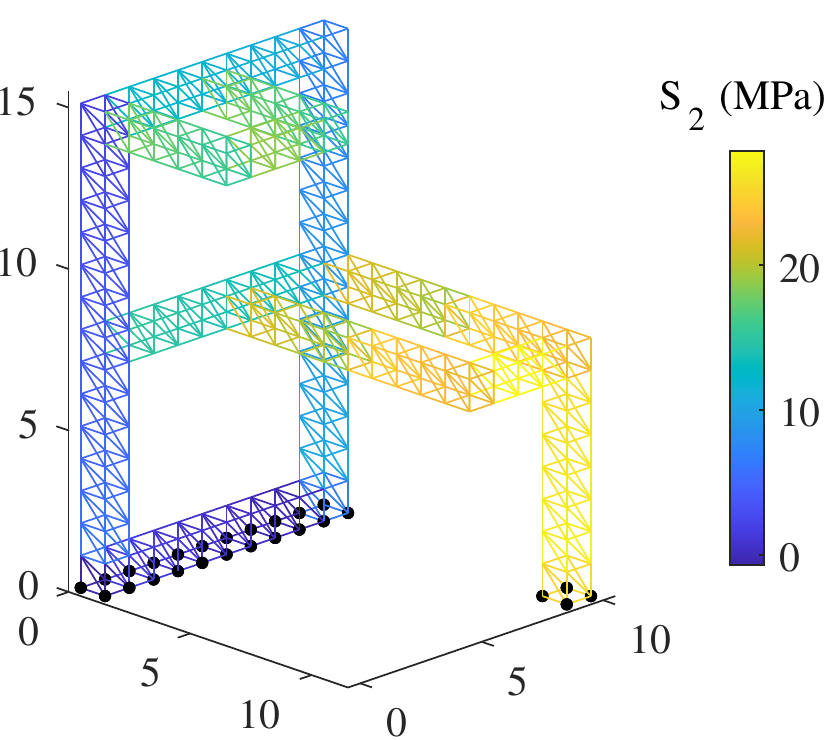}
	\end{subfigure}
    \caption{3D-printed polymeric truss. Displacements and stresses (color) for applied frequency $\Omega = 1000 \text{Hz}$. a) Real component of displacements and stresses; b) Imaginary component of displacements and stresses. } \label{x7mbG0}
\end{center}
\end{figure}

We envision a truss structure such as shown in Fig.~\ref{T05Xnh}a, comprised of 376 nodes and 1,246  elements made of polyurea. The truss is supported at the nodes marked by black dots in the figure. The displacements of a subset of the nodes, marked by arrows in the figure, undergo forced harmonic excitation at frequency $\Omega = 1000 \text{Hz}$. The material is modelled as a linear viscoelastic material and is assumed to be exactly characterized by the Prony series measures experimentally by Knauss {\sl et al.} \cite{knauss2007improved}, Fig.~\ref{T05Xnh}b. The reference solution for those data is shown in Fig.~\ref{x7mbG0}.

\begin{figure}[ht]
\begin{center}
	\begin{subfigure}{0.45\textwidth}\caption{} \includegraphics[width=0.99\linewidth]{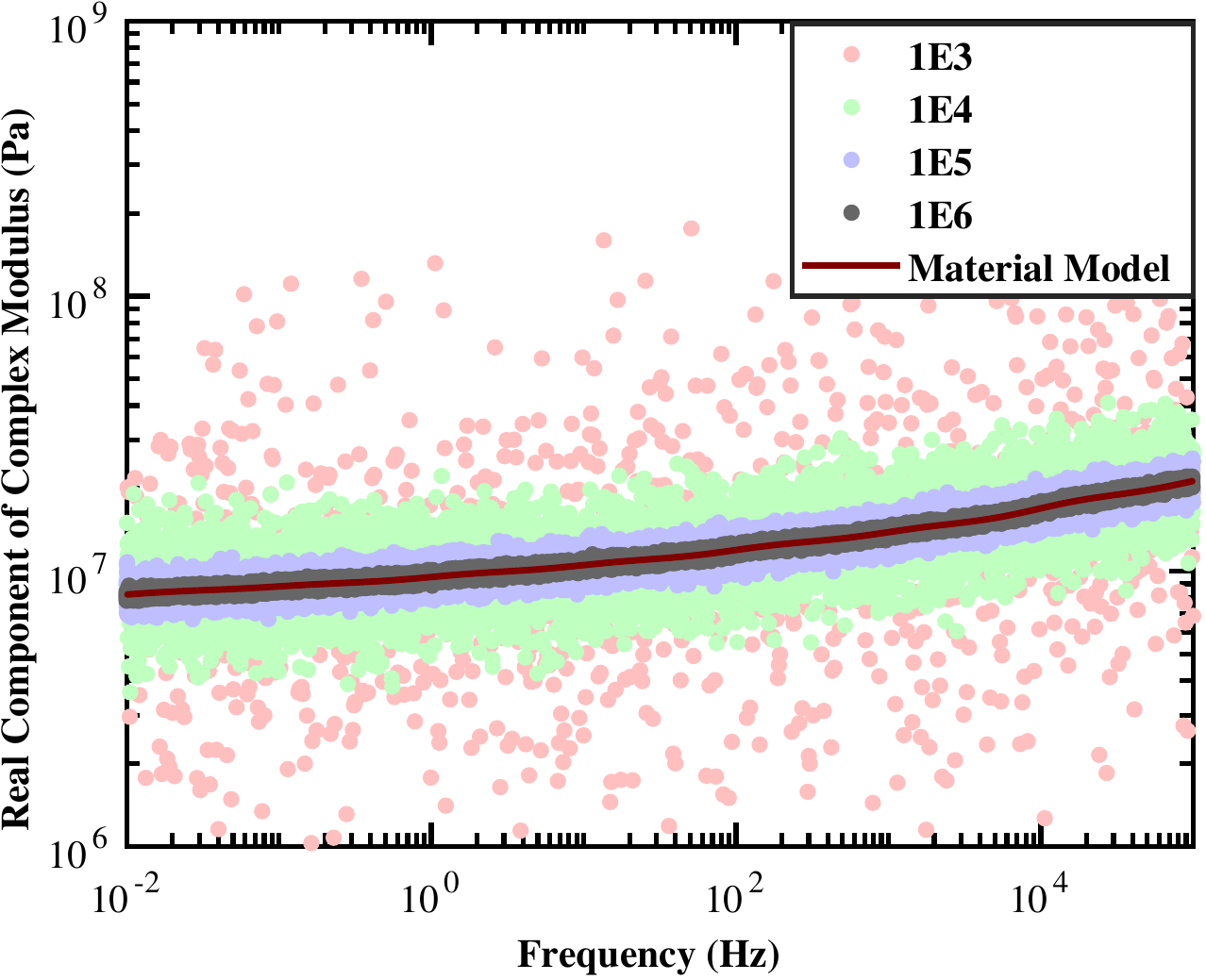}
	\end{subfigure}
$\quad$
	\begin{subfigure}{0.45\textwidth}\caption{} \includegraphics[width=0.99\linewidth]{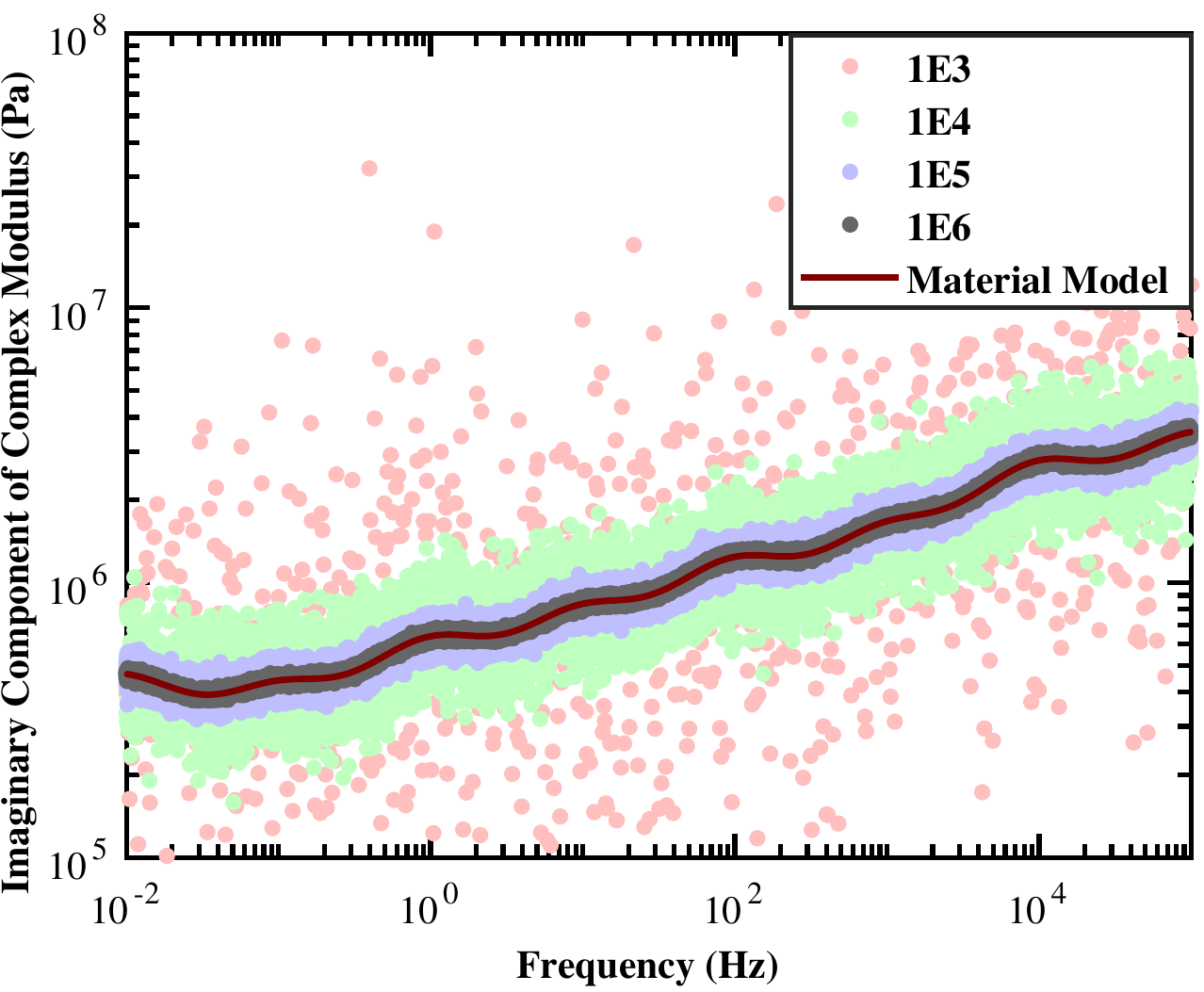}
	\end{subfigure}
    \caption{3D-printed polymeric truss. Point data sets $\mathcal{P}_h$ of sizes $10^3, 10^4, 10^5$ and $10^6$ used in the DD calculations. a) Real component of complex modulus. b) Imaginary component of complex modulus. } \label{y6WBci}
\end{center}
\end{figure}

Suppose that the exact characterization of the material, assumed to be as shown in Fig.~\ref{T05Xnh}b, is not known, but, instead, the material is characterized by a sequence of point data sets $\mathcal{P}_h$ that sample the exact material behavior with increasing fidelity. In practice, the data could be the result of increasingly extensive successive experimental campaigns, or could be the result of diverse experimental campaigns of different degrees of reliability. The noise in the data could be due to experimental measurement error, variability of the material samples, or combinations thereof. The specific data sets used in calculations are shown in Fig.~\ref{y6WBci}. The data sets are synthetically generated by adding random noise to the true behavior of Fig.~\ref{T05Xnh}b. A steady increase in fidelity as achieved by: i) increasing the size of the sample, and ii) simultaneously decreasing the standard deviation of the noise. This type of data convergence has been termed {\sl uniform} in \cite{KO16, Conti:2018}. The data-set sizes used in calculations are $10^3, 10^4, 10^5$ and $10^6$.

\begin{figure}
    \centering
    \includegraphics[scale=0.65]{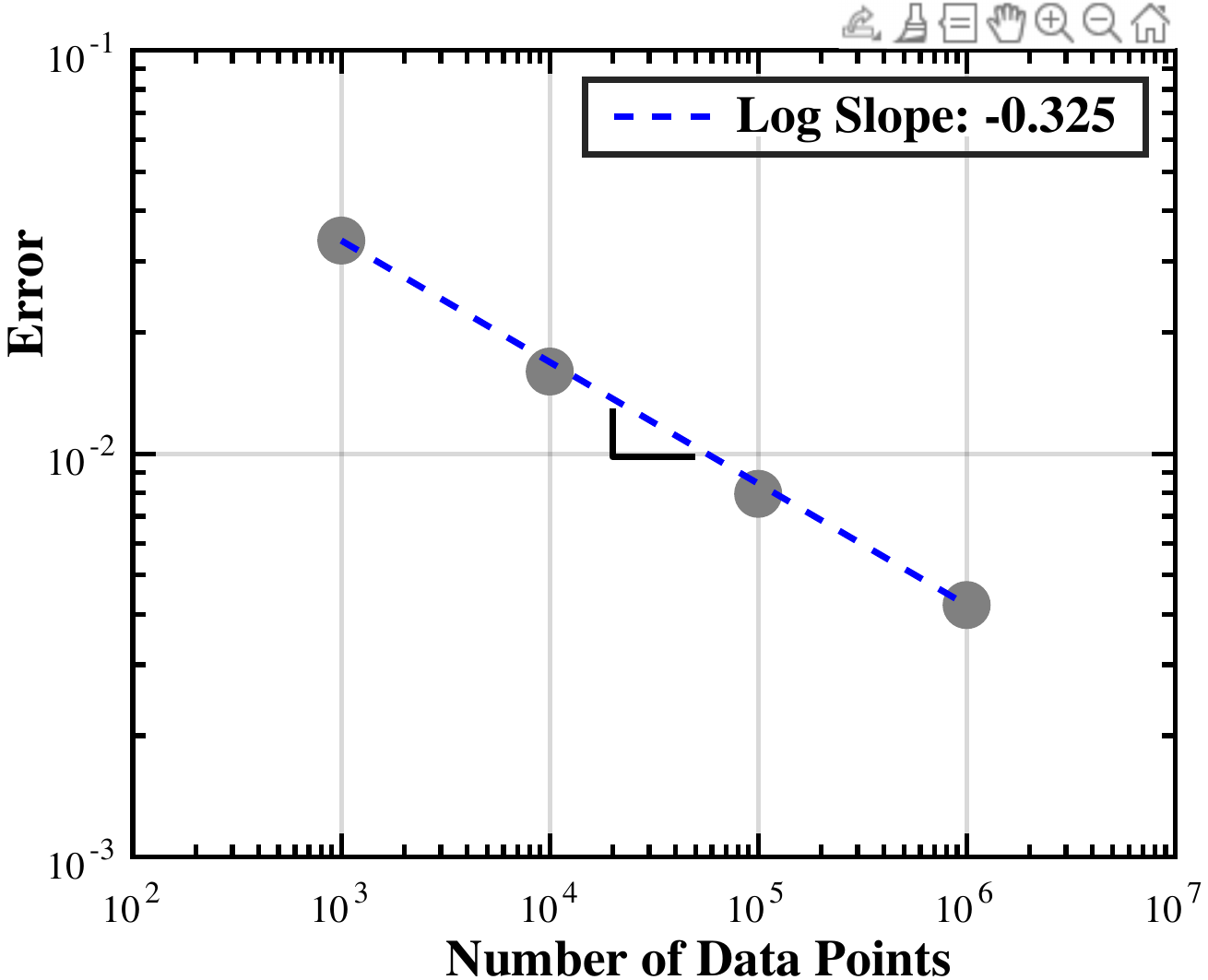}
    \caption{3D-printed polymeric truss. Normalized flat-norm convergence error of the DD solutions relative to the exact reference solution as a function of the number of data points, showing a clear trend towards convergence. }
    \label{2ZG9yT}
\end{figure}

The truss problem is solved for each of the point data sets $\mathcal{P}_h$ using the fixed-point iteration solver of Algorithm~\ref{9WfJmZ}. The fast convergence of the iterative solver, which is attained in ten or fewer iterations, is remarkable given the combinatorial complexity of the data assignment operation and renders the calculations extremely fast. The normalized flat-norm error in the DD solutions, relative to the exact reference solution, is shown in Fig.~\ref{2ZG9yT} as a function of the point-data size. As is evident from the figure, the DD solutions exhibit a clear trend towards convergence at a rate of $\sim 1/3$, which illustrates the convergence property of the general scheme.

\subsection{Insonated gel block}

\begin{figure}[ht]
\begin{center}
	\begin{subfigure}{0.445\textwidth}\caption{} \includegraphics[width=0.99\linewidth]{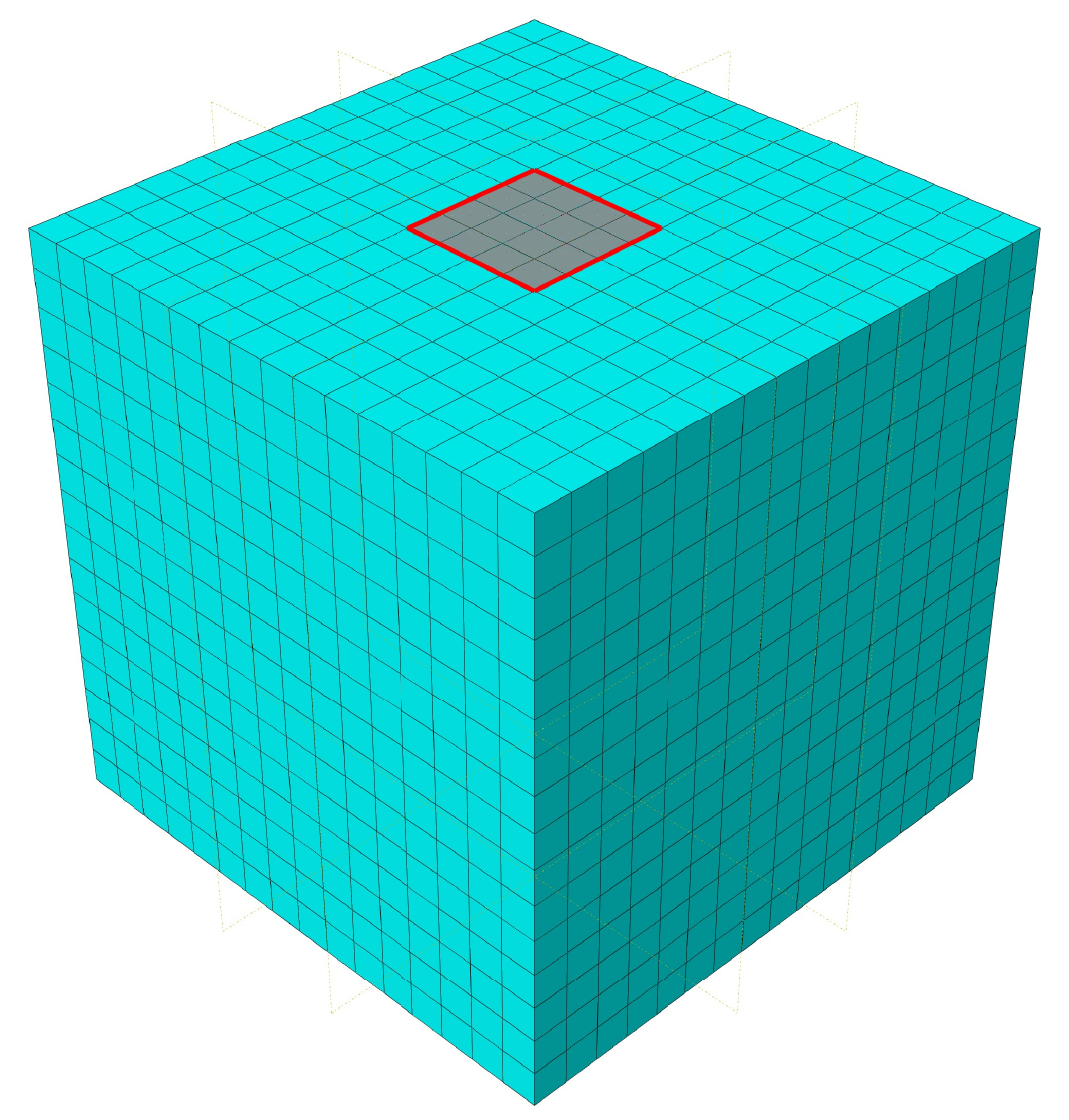}
	\end{subfigure}
	\begin{subfigure}{0.545\textwidth}\caption{} \includegraphics[width=0.99\linewidth]{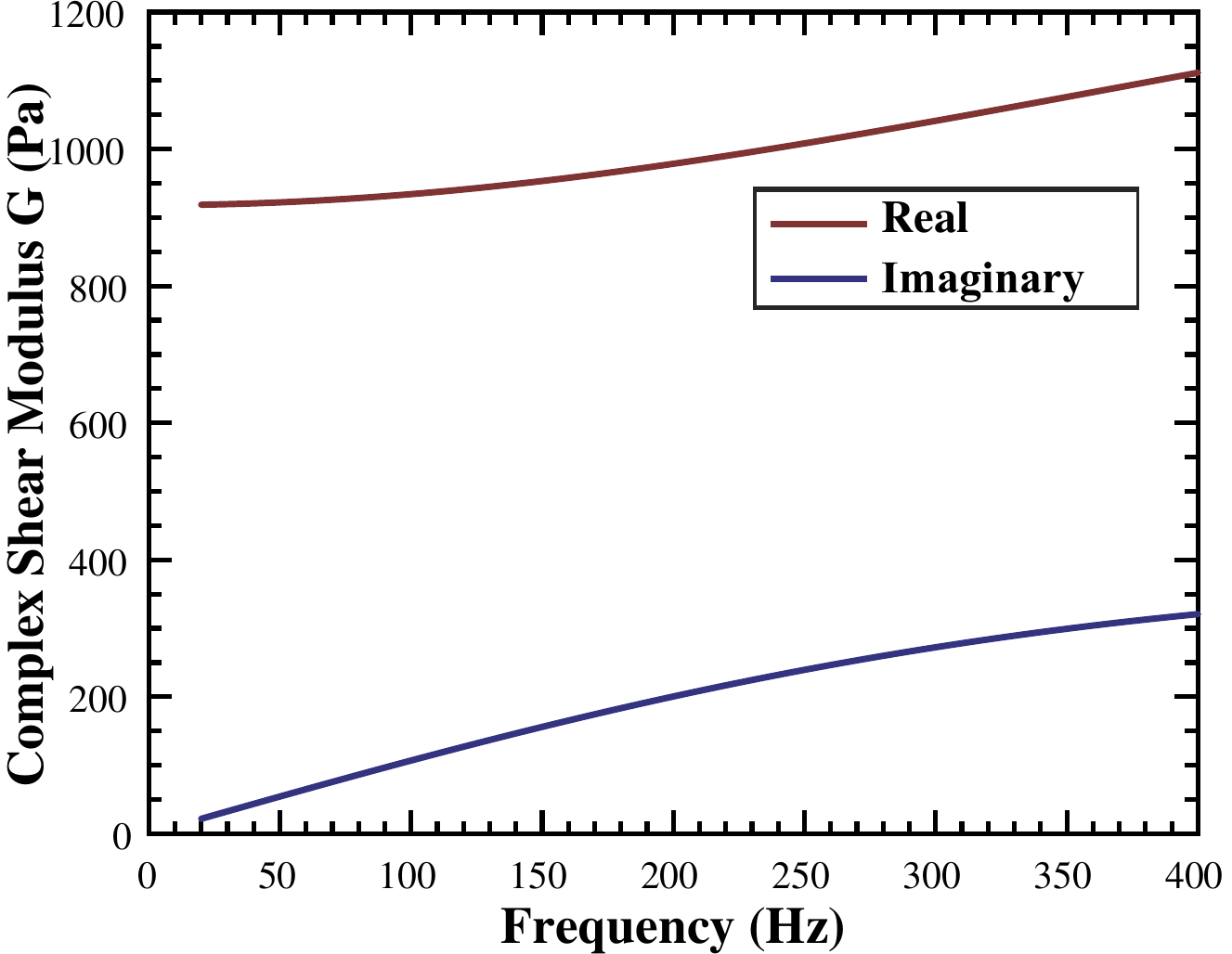}
	\end{subfigure}
    \caption{Insonated gel block. a) Finite-element discretization of gel block into hexahedral elements. The square on the top surface marks the area of application of the ultrasound pressure. b) Complex moduli of agarose gel measured using dynamic shear testing (DST) and magnetic resonance elastography (MRE) (adapted from \cite{okamoto2011viscoelastic}).} \label{TytA4e}
\end{center}
\end{figure}

\begin{figure}[ht]
\begin{center}
	\begin{subfigure}{0.45\textwidth}\caption{} \includegraphics[width=0.99\linewidth]{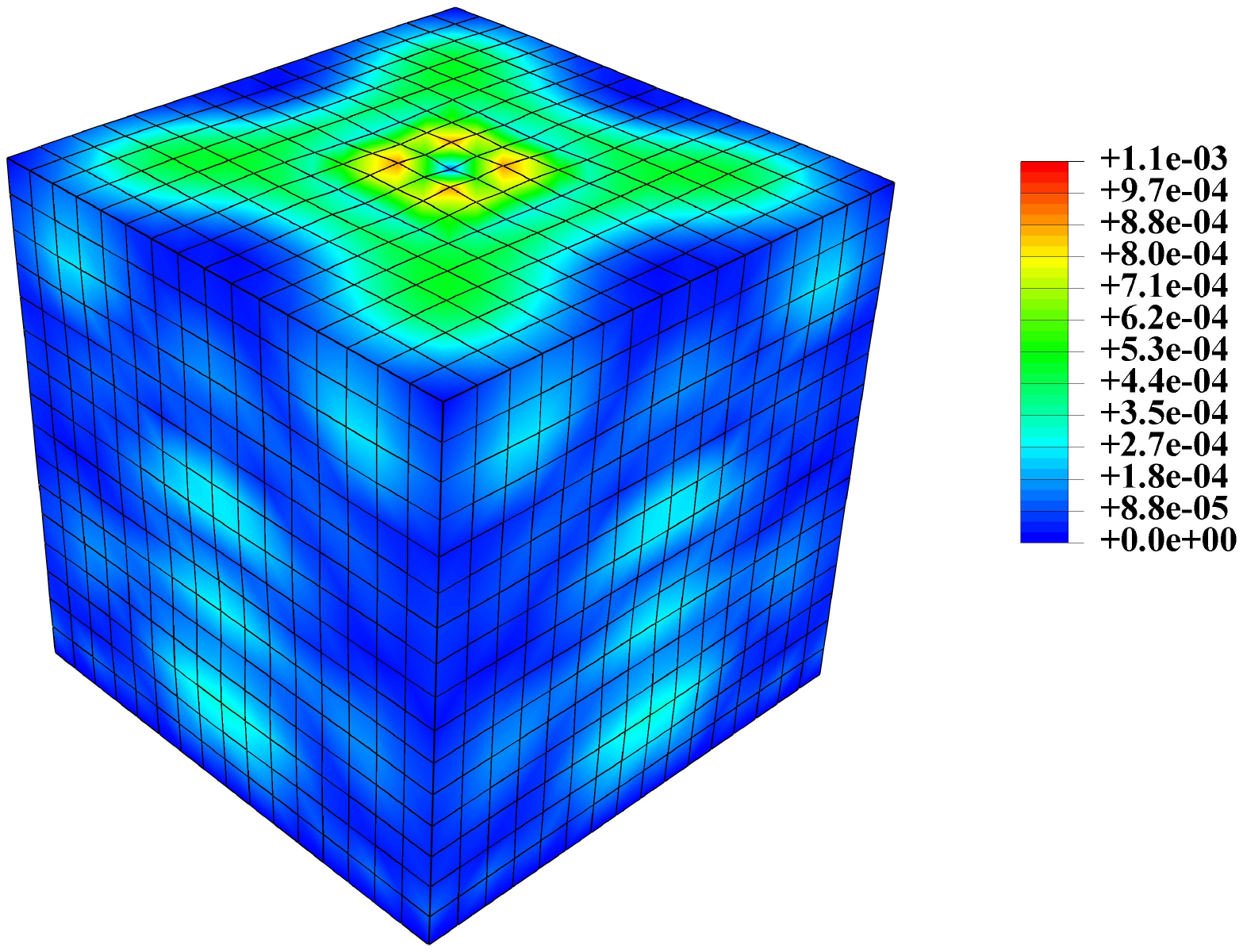}
	\end{subfigure}
$\quad$
	\begin{subfigure}{0.45\textwidth}\caption{} \includegraphics[width=0.99\linewidth]{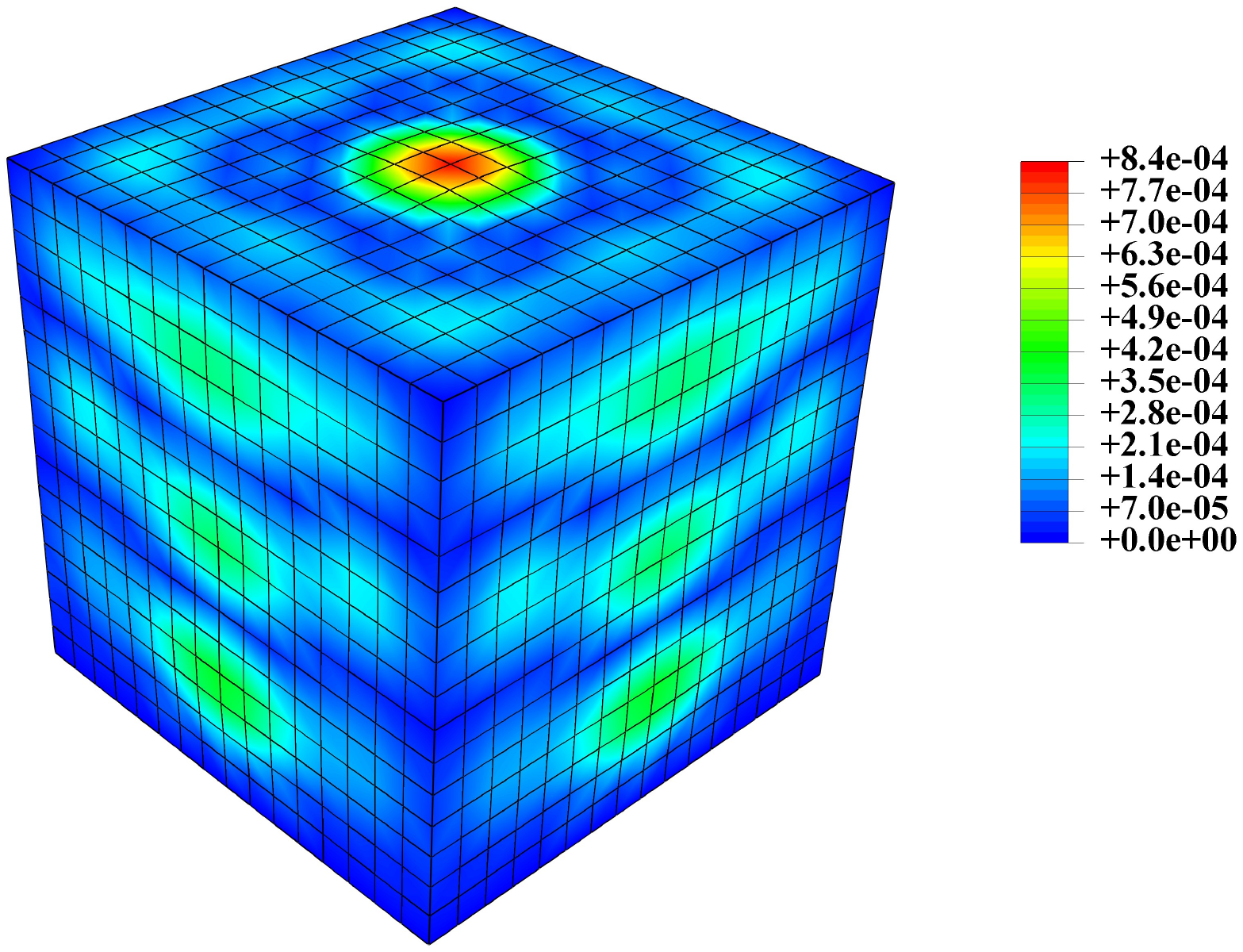}
	\end{subfigure}
    \caption{Insonated gel block. Displacements and stresses (color) for applied frequency $\Omega = 1000 \text{Hz}$. a) Real component of displacements and stresses; b) Imaginary component of displacements and stresses. } \label{sRu8Br}
\end{center}
\end{figure}

We demonstrate the ease of implementation of the DD approach in conventional finite-element codes by performing DD finite-element calculations using the commercial code Abaqus (Abaqus/Standard, Dassault Systemes Simulia, France). The problem concerns the response of a sample of agarose gel in the shape of a block under the action of sonic excitation on its boundary. Agarose gel is often used as a phantom material for soft biological tissue \cite{Rabut:2020}. The analysis is structured as that of the preceding section.

Fig.~\ref{TytA4e}a shows the finite element model of a block 8 cm in size. The model is comprised of 4,913 nodes and 4,096 8-node hexahedral elements. A $4$ ${\rm cm}^2$ square on the top surface is subject to sonic pressure of 1 kPa amplitude and frequency $\Omega = 100 Hz$, while the bottom surface of the block is fully constrained. The material is assumed to be linear viscoelastic and to be characterized by experimental data obtained using dynamic shear testing (DST) and magnetic resonance elastography (MRE) adapted from \cite{okamoto2011viscoelastic}, Fig.~\ref{TytA4e}b. The solution for the assumed exact material characterization is shown in Fig.~\ref{sRu8Br}.

\begin{figure}[ht]
\begin{center}
	\begin{subfigure}{0.45\textwidth}\caption{} \includegraphics[width=0.99\linewidth]{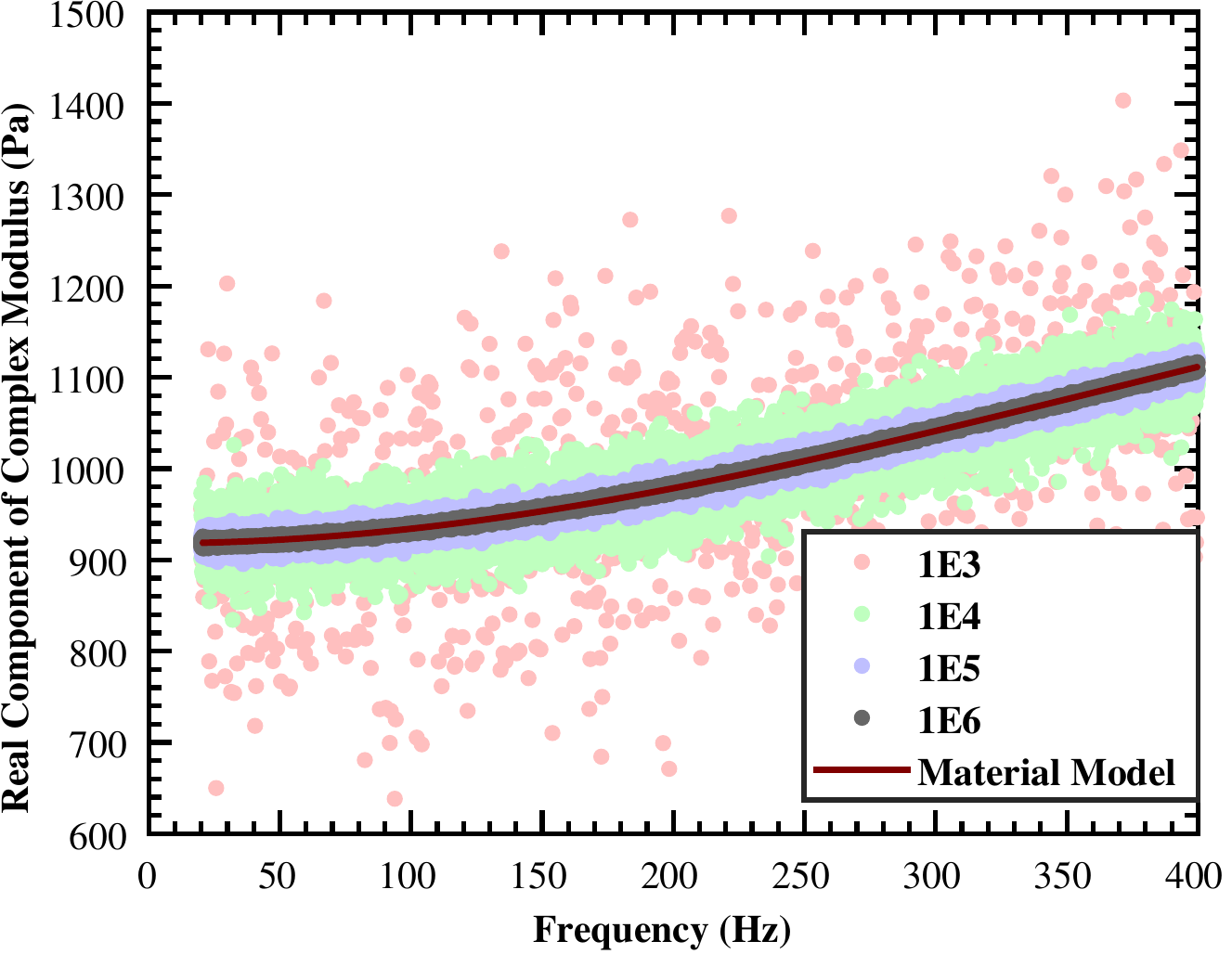}
	\end{subfigure}
$\quad$
	\begin{subfigure}{0.45\textwidth}\caption{} \includegraphics[width=0.99\linewidth]{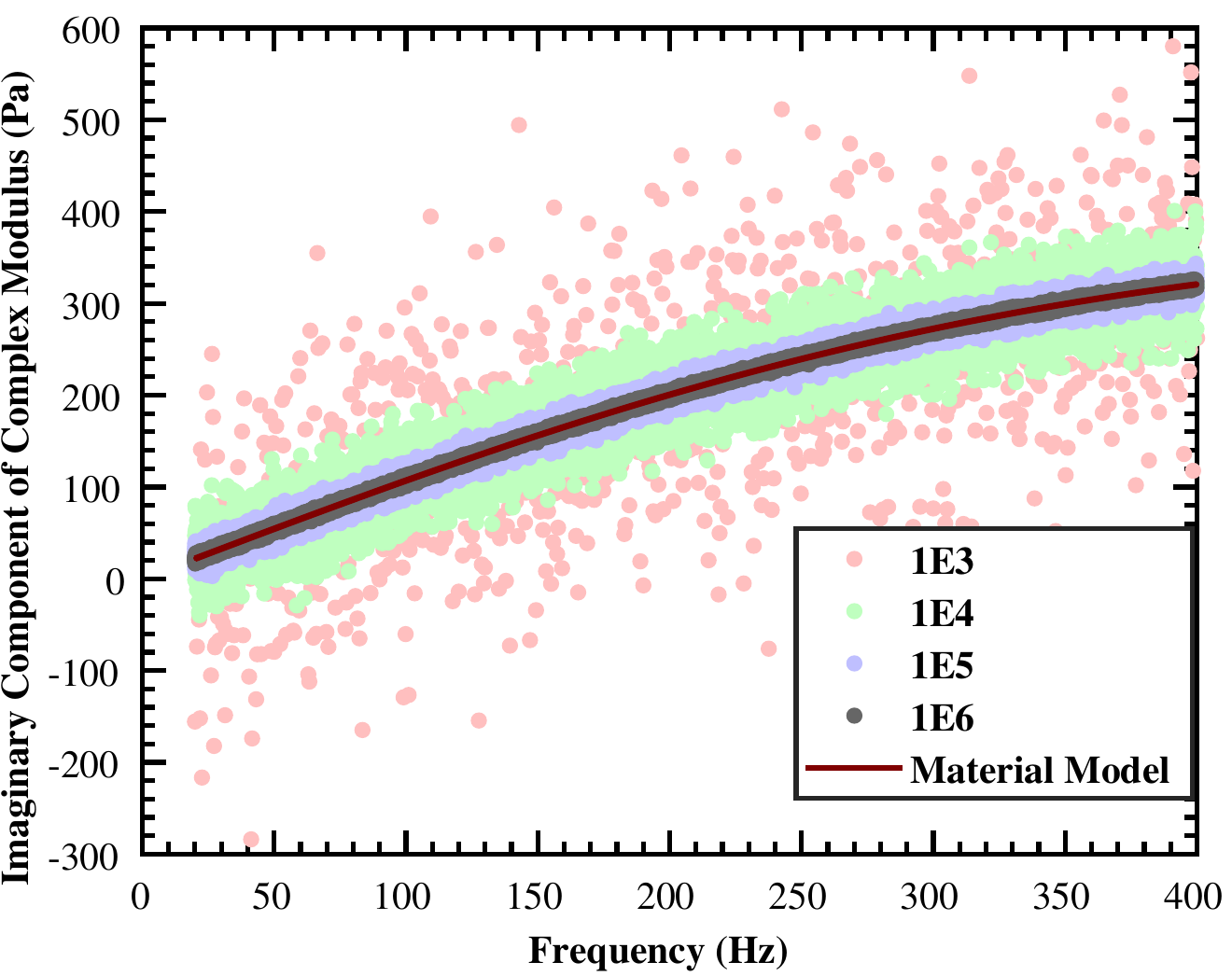}
	\end{subfigure}
    \caption{Insonated gel block. Point data sets $\mathcal{P}_h$ of sizes $10^3, 10^4, 10^5$ and $10^6$ used in the DD calculations. a) Real component of complex modulus. b) Imaginary component of complex modulus. } \label{i7bTVz}
\end{center}
\end{figure}

As in the preceding section, we assume next that the exact characterization of the material is not known exactly, but only through a sequence of point-data sets $\mathcal{P}_h$ of increasing fidelity, Fig.~\ref{i7bTVz}. The data is generated synthetically from the exact characterization by the addition of random noise, representing experimental scatter, measuring error, sample-to-sample variability, multiple provenances, or other sources of error. The data sets contain $10^3, 10^4, 10^5$ and$, 10^6$ points, respectively, and, simultaneously, exhibit decreasing amounts of scatter. In order to assess convergence, the normalized flat-norm error of the DD solutions relative to the limiting reference solution is plotted against the point-data set size, Fig.~\ref{TVd99l}. The robust trend to convergence of the DD solutions, with convergence rate $\sim 1/2$, is evident from the figure. Again remarkably, the fixed-point iteration solver converges in ten iterations or less, which renders the calculations exceedingly fast.

\begin{figure}
    \centering
    \includegraphics[scale=0.65]{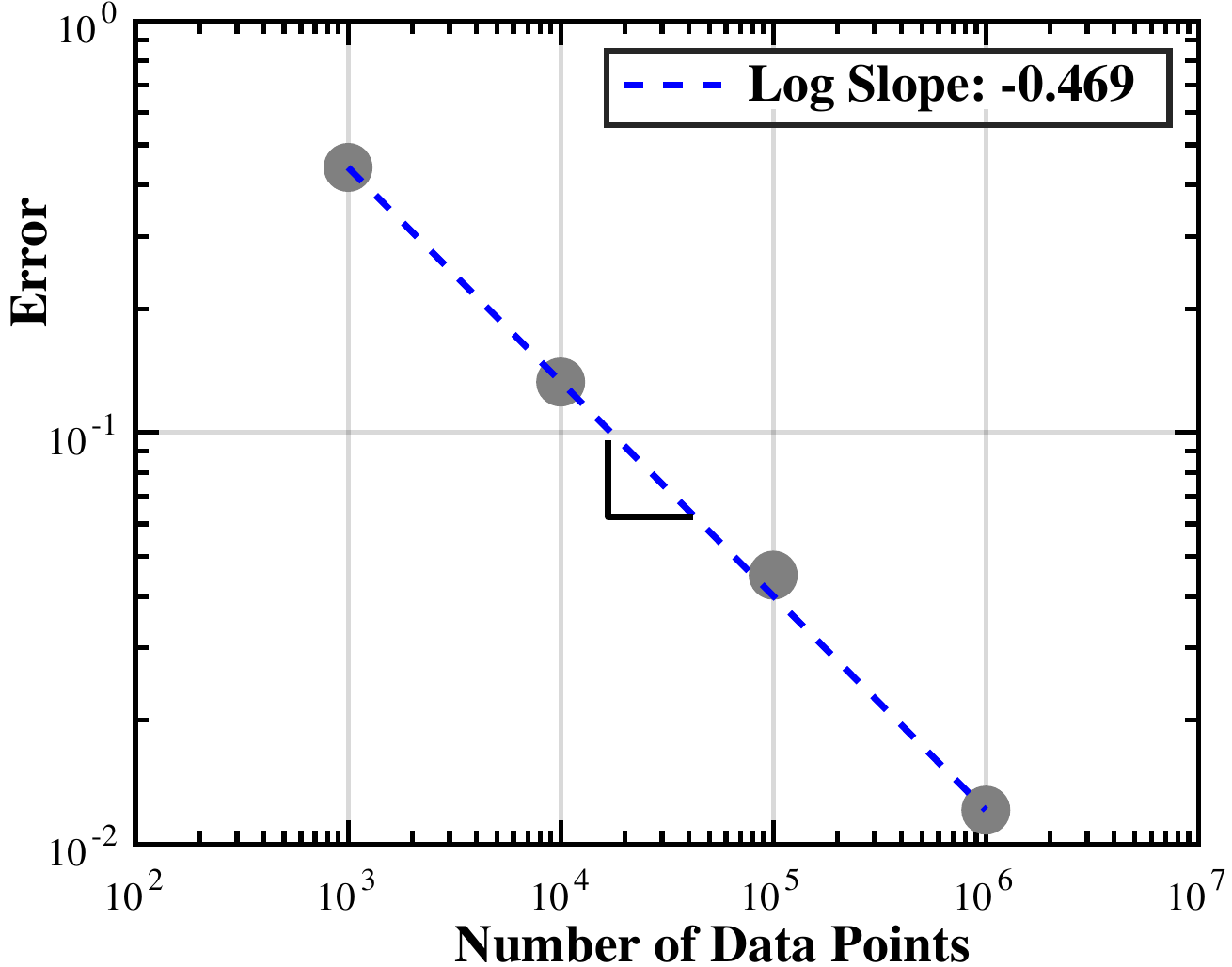}
    \caption{Insonated gel block. Normalized flat-norm convergence error of the DD solutions relative to the exact reference solution as a function of the number of data points, showing a clear trend towards convergence. }
    \label{TVd99l}
\end{figure}

\section{Concluding remarks}\label{Br1zak}

The present work showcases one particular way in which history data can be measured, codified and then used in model-free Data-Driven (DD) calculations, namely, through the use of the Fourier transform. This representational paradigm for histories is particularly well-suited in connection with long-term solid and structural dynamics where the histories of interest often take the form of 'infinite wave trains', a harmonic function being the simplest case in point \cite{MO04}. Other representational paradigms for histories better suited for short-term behavior of materials with memory, such as viscoplastic materials, have been treated elsewhere \cite{Eggers:2019}. It is therefore convenient that a number of experimental techniques supply data pertaining to the Fourier representation, including Dynamic Mechanical Analysis (DMA), nano-indentation, Dynamic Shear Testing (DST) and Magnetic Resonance Elastography (MRE), among others.

In the context of DD methods, in particular, it becomes necessary to measure distances between trial histories, satisfying compatibility and equilibrium, and material histories codified in material-data repositories. The Fourier transforms of such 'infinite histories' are localized measures in Fourier space, a Dirac-delta being the simplest case in point. There are a number of natural ways to measure distance between such measures, most of them variants of the Wasserstein distance from optimal transport \cite{CV08}, which quantify both the difference in amplitude and in the 'center of mass' of the measures. We have chosen to measure distance in Fourier space by means of the flat norm and shown that it lends itself to a straightforward implementation while providing the required control over the approximating solutions. Another advantage of the flat norm in Fourier space is that it allows data at different frequencies to be 'mixed and matched'. This feature is particularly useful when the data is sparse and specific frequencies are sampled only sparingly.

Finally, we emphasize that all DD calculations presented in this work have been directly predicated on  material data without any recourse to material modeling of any type at any stage of the calculations. The ability to make quantitative---and convergent---predictions directly from material data is remarkable and sets forth a new model-free paradigm in computational mechanics.

\section*{Acknowledgements}

Support for this work from the U.~S.~National Institutes of Health through Grant No. 1RF1MH117080 is gratefully acknowledged. MO is also grateful for support from the Deutsche Forschungsgemeinschaft (DFG, German Research Foundation) {\sl via} project 211504053 - SFB 1060; project 441211072 - SPP 2256; and project 390685813 -  GZ 2047/1 - HCM.

\begin{appendix}

\section{The flat norm}\label{af7qTx}

We recall that the flat norm of a vector-valued Radon measure $\mu$ over $\mathbb{R}$ is (cf., e.~g., \cite{CD21}),
\begin{equation}\label{23vzTp}
\begin{split}
    &
    \| \mu \|
    = \\ &
    \sup
    \Big\{
        | \int_{\mbb{R}} f \cdot d\mu | \, : \,
        f:\mbb{R}\to \mathbb{C}^n \ {\rm Lipschitz},
        \ {\rm Lip} \, f \leq \lambda,\ \sup \|f\| \leq 1
    \Big\} ,
\end{split}
\end{equation}
for some cutoff $\lambda > 0$. From the definition, it follows that (cf.~\cite{CD21}, Remark 1.34),
\begin{equation}
    \| A \delta_\alpha - B \delta_\beta \|
    =
    \min\{\|A\|,\|B\|\} \, \min\{2,\lambda |\alpha-\beta|\} + \|A - B\| ,
\end{equation}
where $A, B \in \mathbb{C}^n$, $A\neq 0$, $B\neq 0$, $A-B\neq 0$, $\alpha, \beta \in \mathbb{R}$.

Consider now the following problem. For a given vector-valued Dirac measure $B\delta_\beta$ over $\mathbb{R}$,  $B \in \mathbb{C}^n$, $\beta \in \mathbb{R}$, we wish to find the distance to a linear subspace of measures
\begin{equation}
    \mathcal{S}
    =
    \{ A \delta_\alpha \, : \, A \in S \}
    =
    S \delta_\alpha,
\end{equation}
where $\alpha \in \mathbb{R}$ is given and $S$ is a proper subspace of $\mathbb{C}^n$. Thus, we wish to find
\begin{equation}
    d(B\delta_\beta, \mathcal{S})
    =
    \min
    \{
        \|A\delta_\alpha - B\delta_\beta\| \,:\, A \in S
    \} .
\end{equation}
We shall assume that the norm $\|\cdot\|$ over $\mathbb{C}^n$ derives from a Hermitian inner product. Therefore, every vector $B \in \mathbb{C}^n$ admits an orthogonal decomposition of the form
\begin{equation}\label{6NJwsx}
    B = B_\parallel + B_\perp ,
\end{equation}
where $B_\parallel \in S$ and $B_\perp$ in the orthogonal complement $S_\perp$.

We state the result as a proposition for ease of reference.

\begin{prop}\label{iWC5ll}
Let $B \in \mathbb{C}^n$, $\alpha, \beta \in \mathbb{R}$. Let $S$ be a proper subspace of $\mathbb{C}^n$ and $\mathcal{S} = \{ A \delta_\alpha \, : \, A \in S \}$. Suppose that $\mathbb{C}^n$ is metrized by a norm $\|\cdot\|$ derived from a Hermitian inner product. Then,
\begin{equation}\label{evk9G5}
    d(B\delta_\beta, \mathcal{S})
    =
    \min
    \{
        \|(\mu-1)B_\parallel-B_\perp\| + c \, \mu \, \|B_\parallel\| ,
        \ \| B \|
    \} ,
\end{equation}
with
\begin{equation}\label{70Shtr}
    \mu
    =
    1
    -
    \frac{c }{\sqrt{1-c^2}}
    \frac{\|B_\perp\|}{\|B_\parallel\|} ,
    \quad
\end{equation}
and
\begin{equation}\label{bAF0Ib}
    c = \min\{2,\lambda |\alpha-\beta|\} .
\end{equation}
\end{prop}

\begin{proof}
Suppose that the distance-minimizing vector $A$ is such that
\begin{equation}\label{K9Oi0O}
    \|A \| \leq \|B\|
\end{equation}
and, therefore,
\begin{equation}
    \| A \delta_\alpha - B \delta_\beta \|
    =
    \|A-B\| + c \, \|A\| ,
\end{equation}
with $c$ as in (\ref{bAF0Ib}). Inserting the orthogonal decomposition (\ref{6NJwsx}), we obtain
\begin{equation}
    \| A \delta_\alpha - B \delta_\beta \|
    =
    \|A-B_\parallel-B_\perp\| + c \, \|A\| .
\end{equation}
By orthogonality, with $A \in S$,
\begin{equation}\label{zbEj6t}
    \| A \delta_\alpha - B \delta_\beta \|
    =
    \Big(\|A-B_\parallel\|^2+\|B_\perp\|^2\Big)^{1/2} + c \, \|A\| .
\end{equation}
Differentiating with respect to $A$, we find
\begin{equation}
    \frac{A-B_\parallel}{\Big(\|A-B_\parallel\|^2+\|B_\perp\|^2\Big)^{1/2}}
    +
    c \,
    \frac{A}{\|A\|}
    =
    0 ,
\end{equation}
which shows that
\begin{equation}
    A = \mu B_\parallel ,
\end{equation}
for some $\mu \geq 0$. Inserting this representation into (\ref{zbEj6t}) gives
\begin{equation}
    \| A \delta_\alpha - B \delta_\beta \|
    =
    \|(\mu-1)B_\parallel-B_\perp\| + c \, \mu \, \|B_\parallel\| ,
\end{equation}
and differentiating with respect to $\mu$, we obtain
\begin{equation}
    \frac{((\mu-1)B_\parallel-B_\perp)\cdot B_\parallel}{\|(\mu-1)B_\parallel-B_\perp\|}
    +
    c \, \|B_\parallel\|
    =
    0,
\end{equation}
Solving fo $\mu$ gives (\ref{70Shtr}), whence (\ref{evk9G5}) follows. We verify that
\begin{equation}
    \|A\|
    =
    \mu \|B_\parallel\|
    =
    \Big(
        1
        -
        \frac{c }{\sqrt{1-c^2}}
        \frac{\|B_\perp\|}{\|B_\parallel\|}
    \Big)
    \|B_\parallel\|
    \leq
    \|B_\parallel\|
    \leq
    \| B \| ,
\end{equation}
which bears out the assumption (\ref{K9Oi0O}). The condition
\begin{equation}
    \mu
    \geq
    1
    -
    \frac{c }{\sqrt{1-c^2}}
    \frac{\|B_\perp\|}{\|B_\parallel\|}
    \geq
    0
\end{equation}
additionally requires that
\begin{equation}
    c
    \leq
    \frac{\|B_\parallel\|}{\|B\|} .
\end{equation}
Otherwise, the minimizer is $A=0$ and the corresponding distance
\begin{equation}
    \| A \delta_\alpha - B \delta_\beta \|
    =
    \|B\| ,
\end{equation}
which completes (\ref{evk9G5}).
\end{proof}
\end{appendix}

\bibliographystyle{elsarticle-num}
\bibliography{DDVisco}

\begin{thebibliography}{10}
\expandafter\ifx\csname url\endcsname\relax
  \def\url#1{\texttt{#1}}\fi
\expandafter\ifx\csname urlprefix\endcsname\relax\def\urlprefix{URL }\fi
\expandafter\ifx\csname href\endcsname\relax
  \def\href#1#2{#2} \def\path#1{#1}\fi

\bibitem{PAN21}
S.~Patra, P.~M. Ajayan, T.~N. Narayanan, Dynamic mechanical analysis in
  materials science: The novice’s tale, Oxford Open Materials Science 1~(1)
  (2021) itaa001.

\bibitem{herbert2008nanoindentation}
E.~Herbert, W.~Oliver, G.~Pharr, Nanoindentation and the dynamic
  characterization of viscoelastic solids, Journal of physics D: applied
  physics 41~(7) (2008) 074021.

\bibitem{herbert2009measuring}
E.~G. Herbert, W.~Oliver, A.~Lumsdaine, G.~M. Pharr, Measuring the constitutive
  behavior of viscoelastic solids in the time and frequency domain using flat
  punch nanoindentation, Journal of materials research 24~(3) (2009) 626--637.

\bibitem{bayly2008magnetic}
P.~Bayly, P.~Massouros, E.~Christoforou, A.~Sabet, G.~Genin, Magnetic resonance
  measurement of transient shear wave propagation in a viscoelastic gel
  cylinder, Journal of the Mechanics and Physics of Solids 56~(5) (2008)
  2036--2049.

\bibitem{arbogast1998material}
K.~B. Arbogast, S.~S. Margulies, Material characterization of the brainstem
  from oscillatory shear tests, Journal of biomechanics 31~(9) (1998) 801--807.

\bibitem{MLRGME95}
R.~Muthupillai, D.~Lomas, P.~Rossman, J.~F. Greenleaf, A.~Manduca, R.~L. Ehman,
  Magnetic resonance elastography by direct visualization of propagating
  acoustic strain waves, science 269~(5232) (1995) 1854--1857.

\bibitem{Xu:2021}
K.~Xu, A.~M. Tartakovsky, J.~Burghardt, E.~Darve, Learning viscoelasticity
  models from indirect data using deep neural networks, Computer Methods in
  Applied Mechanics and Engineering 387 (2021) 114124.

\bibitem{KO16}
T.~Kirchdoerfer, M.~Ortiz, Data-driven computational mechanics, Computer
  Methods in Applied Mechanics and Engineering 304 (2016) 81--101.

\bibitem{Conti:2018}
S.~Conti, S.~M{\"u}ller, M.~Ortiz, Data-driven problems in elasticity, Archive
  for Rational Mechanics and Analysis 229~(1) (2018) 79--123.

\bibitem{Eggers:2019}
R.~Eggersmann, T.~Kirchdoerfer, S.~Reese, L.~Stainier, M.~Ortiz, Model-free
  data-driven inelasticity, Computer Methods in Applied Mechanics and
  Engineering 350 (2019) 81--99.

\bibitem{Heyden:2016}
S.~Heyden, M.~Ortiz, A.~Fortunelli, All-atom molecular dynamics simulations of
  multiphase segregated polyurea under quasistatic, adiabatic, uniaxial
  compression, Polymer 106 (2016) 100--108.

\bibitem{Manav:2021}
M.~Manav, M.~Ortiz, Molecular dynamics study of the shock response of polyurea,
  Polymer 212 (2021) 123109.

\bibitem{MO04}
S.~M{\"u}ller, M.~Ortiz, On the $\gamma$-convergence of discrete dynamics and
  variational integrators, Journal of Nonlinear Science 14~(3) (2004) 279--296.

\bibitem{CV08}
C.~Villani, Optimal transport: old and new, Vol. 338, Springer, 2009.

\bibitem{Arefin:2021}
A.~M.~E. Arefin, N.~R. Khatri, N.~Kulkarni, P.~F. Egan, Polymer 3d printing
  review: Materials, process, and design strategies for medical applications,
  Polymers 13~(9).

\bibitem{Rabut:2020}
C.~Rabut, S.~Yoo, R.~C. Hurt, Z.~Jin, H.~Li, H.~Guo, B.~Ling, M.~G. Shapiro,
  Ultrasound technologies for imaging and modulating neural activity, Neuron
  108~(1) (2020) 93--110.

\bibitem{knauss2007improved}
W.~G. Knauss, J.~Zhao, Improved relaxation time coverage in ramp-strain
  histories, Mechanics of Time-Dependent Materials 11~(3) (2007) 199--216.

\bibitem{okamoto2011viscoelastic}
R.~J. Okamoto, E.~H. Clayton, P.~V. Bayly, Viscoelastic properties of soft
  gels: comparison of magnetic resonance elastography and dynamic shear testing
  in the shear wave regime, Physics in Medicine \& Biology 56~(19) (2011) 6379.

\bibitem{CD21}
C.~D{\"u}ll, P.~Gwiazda, A.~Marciniak-Czochra, J.~Skrzeczkowski, Spaces of
  measures and their applications to structured population models, Vol.~36,
  Cambridge University Press, 2021.

\end{thebibliography}

\end{document}